\documentclass[a4paper,12pt]{article}
\usepackage{test}
\usepackage[numbers,sort&compress]{natbib}
\usepackage{color}
\usepackage[colorlinks,citecolor=blue,urlcolor=magenta]{hyperref}
\usepackage{doi}

\usepackage[inline,shortlabels]{enumitem}
\setlist[enumerate,1]{label=(\roman*)}

\usepackage[margin = 1.25in]{geometry}
\usepackage{setspace}
\title{Perov's Contraction Principle and Dynamic Programming with Stochastic Discounting}
\author{Alexis Akira Toda\thanks{Department of Economics, University of California San Diego. Email: \href{mailto:atoda@ucsd.edu}{atoda@ucsd.edu}.}}

\newcommand{\cC}{\mathcal{C}}

\begin{document}

\maketitle

\begin{abstract}
This paper shows the usefulness of Perov's contraction principle, which generalizes Banach's contraction principle to a vector-valued metric, for studying dynamic programming problems in which the discount factor can be stochastic. The discounting condition $\beta<1$ is replaced by $\rho(B)<1$, where $B$ is an appropriate nonnegative matrix and $\rho$ denotes the spectral radius. Blackwell's sufficient condition is also generalized in this setting. Applications to asset pricing and optimal savings are discussed.

\medskip

{\bf Keywords:} contraction, dynamic programming, spectral radius, vector-valued metric.


\end{abstract}

\section{Introduction}

Common dynamic programming problems seek to maximize the expected present discounted value of payoffs
\begin{equation}
\E \sum_{t=0}^\infty \beta^tu(x_t,y_t),\label{eq:PDV}
\end{equation}
where $\E$ is the expectation operator, $\beta\in [0,1)$ is the discount factor, $u$ is the flow utility function, and $(x_t,y_t)$ are state and control variables at time $t$. The standard mathematical tools for solving such problems are Banach's contraction mapping theorem and \citeauthor{blackwell1965}'s sufficient conditions \cite{blackwell1965}; see, for instance, \cite{bertsekas-shreve1978,StokeyLucas1989,bertsekas2017dynamic} for some textbook treatment. A key assumption for applying the contraction mapping theorem to dynamic programming is that the discount factor is bounded above by a number strictly less than 1.

Many recent works in economics, however, consider dynamic programming problems with state-dependent discounting that seek to maximize
\begin{equation}
\E \sum_{t=0}^\infty \left(\prod_{s=0}^t \beta_s\right)u(x_t,y_t), \label{eq:PDV_sd}
\end{equation}
where $\beta_t\ge 0$ is the discount factor from time $t-1$ to $t$ (which can be stochastic) and we normalize $\beta_0\equiv 1$; see \cite{krusell-smith1998,Toda2019JME,Cao2020} for several examples and \cite{StachurskiZhang2021} for a review of such models. Note that \eqref{eq:PDV} is a special case of \eqref{eq:PDV_sd} when the discount factor $\beta_t\equiv \beta<1$ is constant. Importantly, in this class of models, it could be $\beta_t\ge 1$ with positive probability, which makes the standard contraction mapping argument inapplicable. Even if the discount factor is constant and strictly less than 1, similar issues arise in some models with ``rate-of-return risk'' \cite{MaStachurskiToda2020JET,MaToda2021JET}; Section \ref{subsec:os} discusses such an example.

This paper shows that, by replacing the (scalar-valued) metric and Banach's contraction mapping theorem by a vector-valued metric and the generalized contraction principle of \citeauthor{Perov1964} \cite{Perov1964}, the standard proof technique becomes applicable to study dynamic programming problems with stochastic discounting. I also present a simple extension of \citeauthor{blackwell1965}'s sufficient conditions in this setting. To illustrate the usefulness of \citeauthor{Perov1964}'s contraction principle, I apply it to an abstract dynamic programming problem with stochastic discounting, an asset pricing model, and an optimal savings problem with rate-of-return risk.

We say that a space $X$ equipped with a vector-valued metric $d:X\times X\to \R_+^I$ is a generalized metric space if $d$ satisfies nonnegativity, symmetry, and entry-wise triangle inequality. Similarly, a generalized contraction can be defined by replacing the modulus $\beta<1$ by a nonnegative matrix $B$ with spectral radius less than 1. \citeauthor{Perov1964} \cite{Perov1964} proved a generalization of the contraction mapping theorem in this setting to study the existence of a solution to a system of ordinary differential equations. See \cite{Zabrejko1997,FilipPetrusel2010} for some reviews and extensions and \cite{Precup2009} for an application to semilinear operator systems.

Reference \cite{StachurskiZhang2021} develops a general theory of dynamic programming with state-dependent discounting by assuming what the authors call the ``eventually discounting'' condition. This condition allows them to show that some iterate of the Bellman operator is a contraction mapping and recover existence and uniqueness results. The approach presented in this paper significantly simplifies the argument in \cite{StachurskiZhang2021}. Instead of showing that some iterate of the Bellman operator is a contraction, my approach clarifies under what conditions a self map behaves like a contraction. For simplicity and readability, this paper only discusses dynamic programming problems with finite state spaces. However, similar results should hold in infinite state spaces by applying the generalization of Perov's contraction mapping theorem discussed in \cite{Zabrejko1997}.

\section{Vector-valued metric and Perov's contraction principle}

Recall that if $(X,d)$ is a complete metric space, a self map $T:X\to X$ is called a contraction mapping with modulus $\beta\in [0,1)$ if
\begin{equation}
d(Tx,Ty)\le \beta d(x,y)\quad \text{for all $x,y\in X$}.\label{eq:contraction}
\end{equation}
The well-known Banach's contraction mapping theorem states that
\begin{enumerate*}
\item a contraction mapping $T$ has a unique fixed point $x^*\in X$ and
\item starting from any initial value $x^0$, the sequence $x^n\coloneqq T^nx^0$ obtained by iterating $T$ converges to $x^*$ at rate $\beta^n$.
\end{enumerate*}

The contraction mapping theorem can be generalized for a vector-valued metric. Let $X$ be a set, $I\in \N$, and $d:X\times X\to \R_+^I$. We say that $d$ is a \emph{vector-valued metric} if the following conditions hold:
\begin{enumerate}
\item\label{item:nonnegative} (Nonnegativity) $d(x,y)\ge 0$ for all $x,y\in X$, with $d(x,y)=0$ if and only if $x=y$,
\item\label{item:symmetry} (Symmetry) $d(x,y)=d(y,x)$ for all $x,y\in X$,
\item\label{item:triangle} (Triangle inequality) $d(x,z)\le d(x,y)+d(y,z)$ for all $x,y,z\in X$.
\end{enumerate}
In condition \ref{item:triangle}, note that for $a=(a_1,\dots,a_I)\in \R^I$ and $b=(b_1,\dots,b_I)\in \R^I$, we write $a\le b$ if and only if $a_i\le b_i$ for all $i$. A set $X$ endowed with a vector-valued metric is called a \emph{vector-valued metric space}.

Let $\norm{\cdot}$ denote the supremum norm on $\R^I$, so $\norm{a}=\max_i\abs{a_i}$ for $a=(a_1,\dots,a_I)\in \R^I$. Note that the supremum norm satisfies the following monotonicity property: if $a,b\in \R^I$ and $0\le a\le b$, then $\norm{a}=\max_i a_i\le \max_i b_i=\norm{b}$. The monotonicity will be repeatedly used in the subsequent discussion. If $(X,d)$ is a vector-valued metric space and we define $\norm{d}:X\times X\to \R_+$ by $\norm{d}(x,y)=\norm{d(x,y)}=\max_i d_i(x,y)$, then $(X,\norm{d})$ is a metric space in the usual sense. To see this, conditions \ref{item:nonnegative} and \ref{item:symmetry} are trivial, and condition \ref{item:triangle} holds because
\begin{align*}
\norm{d}(x,z)&=\norm{d(x,z)}\le \norm{d(x,y)+d(y,z)}\\
&\le \norm{d(x,y)}+\norm{d(y,z)}=\norm{d}(x,y)+\norm{d}(y,z),
\end{align*}
where the first inequality uses condition \ref{item:triangle} for $d$ and the monotonicity of the supremum norm $\norm{\cdot}$. We say that the vector-valued metric space $(X,d)$ is \emph{complete} if the metric space $(X,\norm{d})$ is complete.

We now extend the contraction mapping theorem to vector-valued metric spaces. Below, let $\norm{\cdot}$ also denote the operator norm for $I\times I$ matrices induced by the supremum norm, so $\norm{A}=\sup_{v\in \R^I\backslash\set{0}}\norm{Av}/\norm{v}$. Recall that for a square matrix $A$, the \emph{spectral radius}, denoted by $\rho(A)$, is defined by the largest absolute value of all eigenvalues:
\begin{equation*}
\rho(A)\coloneqq \max\set{\abs{\alpha}:\text{$\alpha$ is an eigenvalue of $A$}}.
\end{equation*}
We introduce the following definition.

\begin{defn}\label{defn:gen_contraction}
Let $(X,d)$ be a vector-valued metric space. A self map $T:X\to X$ is a \emph{generalized contraction} with coefficient matrix $B\ge 0$ if $\rho(B)<1$ and
\begin{equation}
d(Tx,Ty)\le Bd(x,y)\quad \text{for all $x,y\in X$}.\label{eq:contraction_gen}
\end{equation}
\end{defn}

Clearly, \eqref{eq:contraction} is a special case of \eqref{eq:contraction_gen} when $I=1$ and $B=\beta$. We have the following result.

\begin{thm}[\citeauthor{Perov1964}'s Contraction Mapping Theorem \cite{Perov1964}]\label{thm:gen_contraction}
Let $(X,d)$ be a complete vector-valued metric space and $T:X\to X$ be a generalized contraction with coefficient matrix $B=(b_{ij})$ with spectral radius $\rho(B)<1$. Then
\begin{enumerate}
\item $T:X\to X$ has a unique fixed point $x^*\in X$,
\item For any $x^0\in X$, we have $T^nx^0\to x^*$ as $n\to\infty$,
\item For any $\beta\in (\rho(B),1)$, the approximation error $d(T^nx^0,x^*)$ is $O(\beta^n)$.
\end{enumerate}
\end{thm}

Although the proof of Theorem \ref{thm:gen_contraction} is elementary, because it is not easy to find in English, I present it in Appendix \ref{sec:proof}.

As in \cite{blackwell1965}, we can derive a simple sufficient condition for the generalized contraction property \eqref{eq:contraction_gen}. Let $\Omega$ be a nonempty set. For each $i=1,\dots,I$, suppose that $X_i\subset \R^\Omega$ is a subset of all real functions $f_i:\Omega\to \R$ that is a complete metric space with respect to the sup metric
\begin{equation*}
d_i(f_i,g_i)\coloneqq \sup_{x\in\Omega}\abs{f_i(x)-g_i(x)}.
\end{equation*}
Letting $X=\prod_{i=1}^IX_i$ and $d(f,g)=(\dotsc,d_i(f_i,g_i),\dotsc)\in \R_+^I$ for $f=(f_1,\dots,f_I)\in X$ and $g=(g_1,\dots,g_I)\in X$, then $(X,d)$ becomes a complete vector-valued metric space. Furthermore, $X\subset (\R^I)^\Omega$ is partially ordered by letting $f\le g$ if $f_i(x)\le g_i(x)$ for all $i$ and $x\in \Omega$.

\begin{thm}\label{thm:suffcon}
Let $T:X\to X$ be a self map with the following properties:
\begin{enumerate}
\item (Monotonicity) If $f\le g$, then $Tf\le Tg$.
\item (Discounting) There exists a nonnegative matrix $B$ with $\rho(B)<1$ such that for any $f\in X$ and $c\in \R_+^I$, we have $f+c\in X$ and
\begin{equation}
T(f+c)\le Tf+Bc.\label{eq:discounting}
\end{equation}
\end{enumerate}
Then $T$ is a generalized contraction with coefficient matrix $B$.
\end{thm}

\begin{proof}
Take any $f,g\in X$. Then for each $i$ and $x\in \Omega$, we have
\begin{equation*}
f_i(x)\le g_i(x)+\abs{f_i(x)-g_i(x)}\le g_i(x)+d_i(f_i,g_i).
\end{equation*}
Letting $c_i=d_i(f_i,g_i)\in [0,\infty)$ and $c=(c_1,\dots,c_I)\in \R_+^I$, noting that $T$ is monotone and applying \eqref{eq:discounting}, we obtain
\begin{equation*}
(Tf)_i(x)\le (T(g+c))_i(x)\le (Tg)_i(x)+(Bc)_i,
\end{equation*}
so $(Tf)_i(x)-(Tg)_i(x)\le (Bc)_i$. Interchanging the roles of $f,g$, we obtain $(Tg)_i(x)-(Tf)_i(x)\le (Bc)_i$. Therefore $\abs{(Tf)_i(x)-(Tg)_i(x)}\le (Bc)_i$. Taking the supremum over $x\in \Omega$, we obtain $d_i((Tf)_i,(Tg)_i)\le (Bc)_i$, which is the generalized contraction condition \eqref{eq:contraction_gen}.
\end{proof}

\section{Dynamic programming with stochastic discounting}

This section applies \citeauthor{Perov1964}'s Contraction Mapping Theorem (Theorem \ref{thm:gen_contraction}) to solve dynamic programming problems with state-dependent discounting when the exogenous shocks are driven by a finite state Markov chain, which simplifies some of the arguments in \cite{StachurskiZhang2021}. We consider three applications, an abstract dynamic programming problem, an asset pricing model, and an optimal savings problem.

\subsection{Abstract dynamic programming with bounded utility}\label{subsec:DP}

I first consider an abstract model in which the utility function is bounded as in the classical theory of \cite{blackwell1965}.

Let $I=\set{1,\dots,I}$ be a finite set and $\set{i_t}_{t=0}^\infty$ be a Markov chain taking values in $I$ with transition probability matrix $P=(p_{ij})$. Consider the following dynamic programming problem in a Markovian environment. Let $X,Y$ be nonempty sets. At each stage $t=0,1,\dots$, given the exogenous state $i_t\in I$ and endogenous state variable $x_t\in X$, the decision maker chooses the control variable $y_t\in \Gamma_{i_t}(x_t)$, where $\Gamma_i:X\twoheadrightarrow Y$ is a correspondence with $\Gamma_i(x)\neq \emptyset$. Given the exogenous state $i_t$, endogenous state $x_t$, and control $y_t$, the decision maker receives the flow utility $u_{i_t}(x_t,y_t)$ and the next period's state is determined by the law of motion $x_{t+1}=g_{i_ti_{t+1}}(x_t,y_t)$, where $u_i:X\times Y\to [-\infty,\infty)$ and $g_{ij}:X\times Y\to X$. Conditional on transitioning from state $i_t=i$ to $i_{t+1}=j$, the decision maker discounts the next period's flow utility using the discount factor $\beta_{ij}\ge 0$.

Mathematically, the problem is
\begin{align}
&\maximize && \E_0\sum_{t=0}^\infty \left(\prod_{s=0}^t\beta_{i_{s-1}i_s}\right)u_{i_t}(x_t,y_t) \label{eq:dpinf.prob}\\
&\st && (\forall t) y_t\in \Gamma_{i_t}(x_t),~x_{t+1}=g_{i_ti_{t+1}}(x_t,y_t), \notag \\
&&& \text{$x_0\in X$ and $i_0\in I$ given}, \notag 
\end{align}
where $\beta_{i_{-1}i_0}\equiv 1$ and $\Pr(i_{t+1}=j\mid i_t=i)=p_{ij}$.

For a function $V:X\to \R^I$, define the Bellman operator $T$ by
\begin{equation}
(TV)_i(x)=\sup_{y\in \Gamma_i(x)}\set{u_i(x,y)+\sum_{j=1}^Ip_{ij}\beta_{ij}V_j(g_{ij}(x,y))}.\label{eq:dpinf.TV}
\end{equation}
We say that $V^*:X\to \R^I$ satisfies the Bellman equation if $V^*$ is a fixed point of $T$, so $V^*=TV^*$. Standard results \cite{blackwell1965} show that if $\beta_{ij}\equiv \beta<1$ is constant and each $u_i$ is bounded, then $T$ is a contraction mapping (with modulus $\beta$) on the space $(bX)^I$ of bounded functions from $X$ to $\R^I$ (and thus has a unique fixed point $V^*\in (bX)^I$), and that $V^*_{i_0}(x_0)$ is the supremum value of the dynamic programming problem \eqref{eq:dpinf.prob}.

The following theorem generalizes the theory of stochastic dynamic programming to the case with state-dependent discounting.

\begin{thm}\label{thm:dpinf.vstar}
Let $X,Y$ be nonempty sets, $u_i\in b(X\times Y)$, $\Gamma_i:X\twoheadrightarrow Y$ be such that $\Gamma_i(x)\neq\emptyset$ for all $x\in X$, and $g_{ij}:X\times Y\to X$. Define the matrix $B=(p_{ij}\beta_{ij})$ and assume $\rho(B)<1$. Then the followings are true.
\begin{enumerate}
\item $T:(bX)^I\to (bX)^I$ defined by \eqref{eq:dpinf.TV} is a generalized contraction with coefficient matrix $B$.
\item Letting $V^*\in (bX)^I$ be the unique fixed point of $T$, $V^*_{i_0}(x_0)$ is the supremum value of the dynamic programming problem \eqref{eq:dpinf.prob}.
\end{enumerate}
\end{thm}

\begin{proof}
We only show the first claim as the second is similar to standard results.

Since by assumption each $u_i$ is bounded, if $V\in (bX)^I$, then clearly the right-hand side of \eqref{eq:dpinf.TV} is bounded. Therefore $T:(bX)^I\to (bX)^I$. To show that $T$ is a generalized contraction, it suffices to verify the conditions in Theorem \ref{thm:suffcon}. Monotonicity of $T$ is trivial. To show the discounting property \eqref{eq:discounting}, take any $V\in (bX)^I$ and $c\in \R_+^I$. Then $V+c\in (bX)^I$, and it follows from \eqref{eq:dpinf.TV} that
\begin{align*}
(T(V+c))_i(x)&=\sup_{y\in \Gamma_i(x)}\set{u_i(x,y)+\sum_{j=1}^Ip_{ij}\beta_{ij}(V_j(g_{ij}(x,y))+c_j)}\\
&=(TV)_i(x)+\sum_{j=1}^Ip_{ij}\beta_{ij}c_j=(TV)_i(x)+(Bc)_i.
\end{align*}
Since by assumption $\rho(B)<1$, it follows from Theorem \ref{thm:suffcon} that $T$ is a generalized contraction with coefficient matrix $B$.
\end{proof}

\subsection{Asset pricing}\label{subsec:AP}

Consider a financial asset that trades at price $P_t$ and pays dividend $D_t>0$ at time $t$. Standard results in asset pricing \citep{HarrisonKreps1979} show that the absence of arbitrage implies the existence of a stochastic discount factor $M_t$ such that
\begin{equation}
P_t=\E_t[M_{t+1}(P_{t+1}+D_{t+1})].\label{eq:noarbitrage}
\end{equation}

Consider a simple asset pricing model in which $I=\set{1,\dots,I}$ is the finite set of exogenous states and $\set{i_t}_{t=0}^\infty$ is a Markov chain taking values in $I$ with irreducible transition probability matrix $P=(p_{ij})$. Let $m_{ij}>0$ be the stochastic discount factor conditional on transitioning from state $i$ to $j$, that is, $M_{t+1}=m_{ij}$ if $i_t=i$ and $i_{t+1}=j$. Let the dividend growth $G_t\coloneqq D_t/D_{t-1}$ take value $G_{ij}>0$ conditional on transitioning from state $i$ to $j$. Then under what condition does the asset have a finite price-dividend ratio?

Let $v_i=P_t/D_t$ be the price-dividend ratio in state $i$. Dividing both sides of \eqref{eq:noarbitrage} by $D_t>0$, we obtain
\begin{align*}
v_i&=\frac{P_t}{D_t}=\E_t\left[M_{t+1}\frac{D_{t+1}}{D_t}\left(\frac{P_{t+1}}{D_{t+1}}+1\right)\right]\\
&=\sum_{j=1}^Jp_{ij}m_{ij}G_{ij}(v_j+1).
\end{align*}
Defining the vector $v=(v_1,\dots,v_I)$ and matrix $B=(p_{ij}m_{ij}G_{ij})$, the above equation can be written as
\begin{equation*}
v=Tv\coloneqq Bv+B1,
\end{equation*}
where $1=(1,\dots,1)$ is the vector of ones. We can now characterize the price-dividend ratios as follows.

\begin{prop}\label{prop:assetpricing}
The asset has finite price-dividend ratios if and only if $\rho(B)<1$, in which case
\begin{equation}
v=(I-B)^{-1}B1.\label{eq:pdratio}
\end{equation}
\end{prop}

\begin{proof}
Suppose $\rho(B)<1$. Applying Theorem \ref{thm:suffcon}, we can see that $T:\R_+^I\to \R_+^I$ is a generalized contraction with the unique fixed point given by \eqref{eq:pdratio}.

Suppose next that $\rho(B)\ge 1$ and the vector $v\in \R_+^I$ of price-dividend ratios is finite. Multiplying the left Perron vector $u>0$ of $B$ from left to $v=Tv$, it follows from $B1\gg 0$ that
\begin{equation*}
u'v=u'(Bv+B1)=\rho(B)u'v+u'B1\implies 0\ge (1-\rho(B))u'v=u'B1>0,
\end{equation*}
which is a contradiction.
\end{proof}

The asset pricing model is simple enough that Proposition \ref{prop:assetpricing} can be easily proved without appealing to the generalized contraction mapping theorem. However, using Theorem \ref{thm:gen_contraction} clarifies the argument.

\subsection{Optimal savings with rate-of-return risk}\label{subsec:os}

As yet another application and a more concrete example, we consider an optimal savings problem with rate-of-return risk. This problem was recently solved by \cite{MaStachurskiToda2020JET,MaToda2021JET} using the Euler equation approach of \cite{LiStachurski2014}. The proof technique in \cite{MaStachurskiToda2020JET,MaToda2021JET} is to show that some iterate of the time iteration operator $T$ is a contraction. Here I show that $T$ is a generalized contraction in the sense of Definition \ref{defn:gen_contraction}, which significantly simplifies the proof by applying Theorem \ref{thm:gen_contraction}.

Here we briefly describe the problem following \cite{MaToda2021JET}; the reader is referred to \cite{LiStachurski2014,MaStachurskiToda2020JET,MaToda2021JET} for more details. Time is discrete and denoted by $t=0, 1, 2,\dotsc$ Let $a_t\ge 0$ be the financial wealth of the agent at the beginning of period $t$. The agent chooses consumption $c_t\ge 0$ and saves the remaining wealth $a_t-c_t$. The period utility function is denoted by $u$. The discount factor between $t-1$ and $t$, gross return on wealth between $t-1$ and $t$, and income at $t$ are denoted by $\beta_t,R_t,Y_t\ge 0$, where we normalize $\beta_0\equiv 1$. We suppose that these variables are Markov-modulated in the following sense: letting $I=\set{1,\dots,I}$ be a finite set and $\set{i_t}_{t=0}^\infty$ be a Markov chain taking values in $I$ with transition probability matrix $P=(p_{ij})$, we have
\begin{equation}
\beta_t=\beta(i_{t-1},i_t,\zeta_t),\quad R_t=R(i_{t-1},i_t,\zeta_t), \quad
Y_t=Y(i_{t-1},i_t,\zeta_t), \label{eq:betaRY}
\end{equation}
where $\zeta_t$ is an \iid random variable and $\beta,R,Y$ are nonnegative measurable functions. Note that \eqref{eq:betaRY} implies that the discount factor, return on wealth, and income can all depend on the two most recent Markov states $(i_{t-1},i_t)$ as well as the \iid shock $\zeta_t$. Given the initial wealth $a_0=a>0$ and state $i_0$, the agent's objective is to maximize the expected lifetime utility
\begin{equation*}
\E_{i_0}\sum_{t=0}^\infty \left(\prod_{s=0}^t \beta_s\right)u(c_t)
\end{equation*}
subject to the budget constraint
\begin{equation*}
(\forall t)~a_{t+1}=R_{t+1}(a_t-c_t)+Y_{t+1},
\end{equation*}
where consumption satisfies $0\le c_t\le a_t$ (no borrowing). We say that $c(a,i)$ is the \emph{consumption function} if $c_t=c(a_t,i_t)$ solves the optimal savings problem just described.

The idea of Euler equation approach \cite{LiStachurski2014,MaStachurskiToda2020JET,MaToda2021JET} is to update a candidate consumption function using the Euler equation. Namely, let $\cC$ be a space of candidate consumption functions, and consider updating $c(a,i)$ by the unique number $\xi\in [0,a]$ satisfying the Euler equation (first-order condition)
\begin{equation}
u'(\xi)=\min\set{\max\set{\E_i \hat{\beta}\hat{R}u'(c(\hat{R}(a-\xi)+\hat{Y},\hat{i})),u'(a)},u'(0)}.\label{eq:xi}
\end{equation}
(Here variables with hats denote the values next period, for example $i=i_t$ and $\hat{i}=i_{t+1}$; the $\min$ and $\max$ operators take care of the possibility of corner solutions $\xi=0$ and $\xi=a$.) This updating rule defines the time iteration operator $T:\cC\to \cC$ through $(Tc)(a,i)=\xi$, and the unique fixed point of $T$ is the consumption function; see \cite[Section 2]{MaToda2021JET} for details.

In the discussion below, assume the following.
\begin{asmp}\label{asmp:os}
\begin{enumerate*}
\item\label{item:u} The utility function $u:[0,\infty)\to \R \cup \set{-\infty}$ is continuously differentiable on $(0,\infty)$ and $u'$ is positive and strictly decreasing on $(0,\infty)$;
\item\label{item:E} The matrix $B=(b_{ij})$ defined by $b_{ij}=p_{ij}\E[\beta(i,j,\xi)R(i,j,\xi)]$ is finite;
\item\label{item:rho} $\rho(B)<1$.
\end{enumerate*}
\end{asmp}

Assumptions \ref{asmp:os}\ref{item:u}\ref{item:E} guarantee that the time iteration operator $T$ is well-defined \cite[Lemma 1]{MaToda2021JET}. To apply Theorem \ref{thm:gen_contraction}, it is convenient to work with the space of marginal utility functions $f_i(a)=u'(c(a,i))$ instead of consumption functions $c(a,i)$. Thus let $X$ be the space of functions $f:(0,\infty)\to \R_+^I$ such that $a\mapsto f_i(a)$ is continuous, decreasing, and $\sup_{a\in (0,\infty)}\abs{f_i(a)-u'(a)}<\infty$. Define the vector-valued metric $d:X\times X\to \R_+^I$ by
\begin{equation*}
d_i(f,g)=\sup_{a\in (0,\infty)}\abs{f_i(a)-g_i(a)}.
\end{equation*}
Then it is easy to see that $(X,d)$ is a complete vector-valued metric space.

Define $\tilde{T}:X \to X$ by $(\tilde{T}f)_i(a)=u'((Tc)(a,i))$, where $c(a,i)=(u')^{-1}(f_i(a))$. Let us now show that $\tilde{T}$ is a generalized contraction. To this end, we apply the sufficient conditions in Theorem \ref{thm:suffcon}. The monotonicity of $\tilde{T}$ follows from the same argument as \cite[Lemma B.4]{MaStachurskiToda2020JET}. The discounting condition \eqref{eq:discounting} follows from \cite[Lemma 11]{MaToda2021JET}. Finally, by Assumption \ref{asmp:os}\ref{item:rho}, we have $\rho(B)<1$. Therefore $\tilde{T}$ is a generalized contraction, and $\tilde{T}$ has a unique fixed point $f\in X$. We can then recover the consumption function as $c(a,i)=(u')^{-1}(f_i(a))$.

\appendix

\section{Proof of Theorem \ref{thm:gen_contraction}}\label{sec:proof}

Let $\norm{\cdot}$ denote the supremum norm in $\R^I$  as well as the operator norm for $I\times I$ matrices induced by $\norm{\cdot}$. 

Take any $x^0\in X$ and define $x^n=T^nx^0$. Let us first show that $\set{x^n}\subset X$ is bounded in the metric space $(X,\norm{d})$. To see this, iterating \eqref{eq:contraction_gen}, we obtain
\begin{equation*}
d(x^k,x^{k-1})\le B^{k-1}d(x^1,x^0).
\end{equation*}
Summing this inequality over $k=1,\dots,n$ and using the triangle inequality \ref{item:triangle}, we obtain
\begin{equation*}
d(x^n,x^0)\le \sum_{k=1}^nd(x^k,x^{k-1})\le (I+B+\dots+B^{n-1})d(x^1,x^0)
\end{equation*}
for all $n$. Since $B\ge 0$ and $\rho(B)<1$ by assumption, $\sum_{k=1}^n B^{k-1}$ monotonically converges to the nonnegative matrix $\sum_{k=1}^\infty B^{k-1}=(I-B)^{-1}$. Therefore
\begin{equation*}
d(x^n,x^0)\le (I-B)^{-1}d(x^1,x^0)
\end{equation*}
for all $n$. Taking the supremum norm of both sides and using monotonicity, we obtain
\begin{equation*}
\norm{d}(x^n,x^0)=\norm{d(x^n,x^0)}\le \norm{(I-B)^{-1}d(x^1,x^0)}\eqqcolon M<\infty,
\end{equation*}
implying that the sequence $\set{x^n}$ is bounded.

Next let us show that $\set{x^n}$ is a Cauchy sequence in the complete metric space $(X,\norm{d})$ and hence convergent. If $m\ge n$, iterating \eqref{eq:contraction_gen} yields
\begin{equation*}
d(x^m,x^n)\le B^nd(x^{m-n},x^0).
\end{equation*}
Taking the supremum norm of both sides, using monotonicity, and noting that $\set{x^n}$ is bounded, we obtain
\begin{equation*}
\norm{d}(x^m,x^n)=\norm{d(x^m,x^n)}\le \norm{B^n}\norm{d(x^{m-n},x^0)}\le M\norm{B^n}.
\end{equation*}
By the Gelfand spectral radius formula \cite[Theorem 5.7.10]{HornJohnson2013}, we have $\norm{B^n}^{1/n}\to \rho(B)<1$ as $n\to\infty$. Therefore for any $\beta\in (\rho(B),1)$, there exists a constant $C>0$ such that $\norm{B^n}\le C\beta^n$ for all $n$, so $\norm{B^n}\to 0$ and $\set{x^n}$ is Cauchy in $(X,\norm{d})$. Therefore there exists $x^*\in X$ such that $\lim_{n\to\infty} x^n=x^*$.

Let us show that $x^*$ is the unique fixed point of $T$. By the triangle inequality \ref{item:triangle} and \eqref{eq:contraction_gen}, we obtain
\begin{align*}
d(Tx^*,x^*)&\le d(Tx^*,Tx^n)+d(Tx^n,x^*)\\
&\le Bd(x^*,x^n)+d(x^{n+1},x^*).
\end{align*}
Taking the supremum norm of both sides and using monotonicity, we obtain
\begin{equation*}
\norm{d(Tx^*,x^*)}\le \norm{B}\norm{d}(x^n,x^*)+\norm{d}(x^{n+1},x^*)\to 0
\end{equation*}
as $n\to\infty$ because $x^n\to x^*$ in $(X,\norm{d})$. Therefore $d(Tx^*,x^*)=0$ and hence $Tx^*=x^*$, so $x^*$ is a fixed point of $T$. If $x^*,y^*$ are two fixed points, then for any $n$ we have
\begin{equation*}
d(x^*,y^*)=d(T^nx^*,T^ny^*)\le B^nd(x^*,y^*).
\end{equation*}
Taking the supremum norm of both sides, we obtain
\begin{equation*}
\norm{d(x^*,y^*)}\le \norm{B^n}\norm{d(x^*,y^*)}\to 0
\end{equation*}
as $n\to\infty$, so $d(x^*,y^*)=0$ and $x^*=y^*$. Therefore the fixed point is unique.

Finally, for any $x^0$ and $x^n=T^nx^0$, we have
\begin{equation*}
d(x^n,x^*)=d(T^nx^0,T^nx^*)\le B^nd(x^0,x^*).
\end{equation*}
Taking the supremum norm of both sides, we obtain
\begin{equation*}
\norm{d(x^n,x^*)}\le \norm{B^n}\norm{d(x^0,x^*)} \le C\beta^n \norm{d(x^0,x^*)}\to 0,
\end{equation*}
and the approximation error $d(x^n,x^*)$ is $O(\beta^n)$. \qedsymbol


\end{document}